\documentclass{IEEEtran}
\usepackage{soul,times,amssymb,amsthm,amsmath,amsfonts,bbm,mathrsfs,xcolor,subfigure,graphicx,enumitem,booktabs,euscript,nicefrac}

\allowdisplaybreaks
\usepackage{graphicx}
\usepackage[unicode]{hyperref}
\hypersetup{
    colorlinks,
    linkcolor={blue!80!black},
    citecolor={blue!80!black},
    urlcolor={blue!80!black}
}

\usepackage[normalem]{ulem}
\usepackage[utf8]{inputenc}
\usepackage[T1]{fontenc}

\newtheorem{theorem}{Theorem}[section]

\newtheorem{proposition}[theorem]{Proposition}
\newtheorem{corollary}[theorem]{Corollary}
\theoremstyle{remark}
\newtheorem{remark}{Remark}
\newtheorem{definition}{Definition}[section]

\newcommand\nc\newcommand
\nc\bfa{{\boldsymbol a}}\nc\bfA{{\boldsymbol A}}\nc\cA{{\EuScript A}}
\nc\bfb{{\boldsymbol b}}\nc\bfB{{\boldsymbol B}}\nc\cB{{\EuScript B}}
\nc\bfc{{\boldsymbol c}}\nc\bfC{{\boldsymbol C}}\nc\cC{{\mathscr C}}
\nc\bfd{{\boldsymbol d}}\nc\bfD{{\boldsymbol D}}\nc\cD{{\EuScript D}}
\nc\bfe{{\boldsymbol e}}\nc\bfE{{\boldsymbol E}}\nc\cE{{\EuScript E}}
\nc\bff{{\boldsymbol f}}\nc\bfF{{\boldsymbol F}}\nc\cF{{\mathscr F}}
\nc\bfg{{\boldsymbol g}}\nc\bfG{{\boldsymbol G}}\nc\cG{{\EuScript G}}
\nc\bfh{{\boldsymbol h}}\nc\bfH{{\boldsymbol H}}\nc\cH{{\mathcal H}}
\nc\bfi{{\boldsymbol i}}\nc\bfI{{\boldsymbol I}}\nc\cI{{\mathcal I}}
\nc\bfj{{\boldsymbol j}}\nc\bfJ{{\boldsymbol J}}\nc\cJ{{\EuScript J}}
\nc\bfk{{\boldsymbol k}}\nc\bfK{{\boldsymbol K}}\nc\cK{{\EuScript K}}
\nc\bfl{{\boldsymbol l}}\nc\bfL{{\boldsymbol L}}\nc\cL{{\EuScript L}}
\nc\bfm{{\boldsymbol m}}\nc\bfM{{\boldsymbol M}}\nc\cM{{\EuScript M}}
\nc\bfn{{\boldsymbol n}}\nc\bfN{{\boldsymbol N}}\nc\cN{{\EuScript N}}
\nc\bfo{{\boldsymbol o}}\nc\bfO{{\boldsymbol O}}\nc\cO{{\EuScript O}}
\nc\bfp{{\boldsymbol p}}\nc\bfP{{\boldsymbol P}}\nc\cP{{\EuScript P}}
\nc\bfq{{\boldsymbol q}}\nc\bfQ{{\boldsymbol Q}}\nc\cQ{{\EuScript Q}}
\nc\bfr{{\boldsymbol r}}\nc\bfR{{\boldsymbol R}}\nc\cR{{\EuScript R}}
\nc\bfs{{\boldsymbol s}}\nc\bfS{{\boldsymbol S}}\nc\cS{{\EuScript S}}
\nc\bft{{\boldsymbol t}}\nc\bfT{{\boldsymbol T}}\nc\cT{{\EuScript T}}
\nc\bfu{{\boldsymbol u}}\nc\bfU{{\boldsymbol U}}\nc\cU{{\EuScript U}}
\nc\bfv{{\boldsymbol v}}\nc\bfV{{\boldsymbol V}}\nc\cV{{\mathscr V}}
\nc\bfw{{\boldsymbol w}}\nc\bfW{{\boldsymbol W}}\nc\cW{{\mathscr W}}
\nc\bfx{{\boldsymbol x}}\nc\bfX{{\boldsymbol X}}\nc\cX{{\EuScript X}}
\nc\bfy{{\boldsymbol y}}\nc\bfY{{\boldsymbol Y}}\nc\cY{{\mathscr Y}}
\nc\bfz{{\boldsymbol z}}\nc\bfZ{{\boldsymbol Z}}\nc\cZ{{\EuScript Z}}
\nc{\1}{{\mathbbm{1}}}
\nc\rr{{\mathbb R}}
\nc\zz{{\mathbb Z}}
\nc\ee{{\mathbb E}}
\nc\sS{{\mathcal S}}
\nc{\integers}{{\mathbb Z}}
\nc{\ff}{{\mathbb F}}
\nc{\ii}{{\mathbb I}}
\nc{\sC}{{\mathfrak C}}
\nc{\sL}{{\mathfrak L}}
\nc\hH{{\mathsf H}}
\nc\hh{{\mathcal h}}
\nc\gG{{\mathsf G}}

\nc{\remove}[1]{}

\DeclareSymbolFont{bbold}{U}{bbold}{m}{n}
\DeclareSymbolFontAlphabet{\mathbbold}{bbold}

\newcommand{\genstirlingII}[3]{%
  \genfrac{\{}{\}}{0pt}{#1}{#2}{#3}%
}

\newcommand{\stirlingII}[2]{\genstirlingII{}{#1}{#2}}

\nc\Renyi{R{\'e}nyi }
\allowdisplaybreaks	

\usepackage{xargs}
\usepackage[colorinlistoftodos,prependcaption]{todonotes}

\newcommandx{\rednote}[2][1=]{\todo[linecolor=red,backgroundcolor=red!25,bordercolor=red,#1]{#2}}
\newcommandx{\bluenote}[2][1=]{\todo[linecolor=blue,backgroundcolor=blue!25,bordercolor=blue,#1]{#2}}
\newcommandx{\yellownote}[2][1=]{\todo[linecolor=yellow,backgroundcolor=yellow!25,bordercolor=yellow,#1]{#2}}
\newcommandx{\greennote}[2][1=]{\todo[inline,linecolor=olive,backgroundcolor=green!25,bordercolor=olive,#1]{#2}}

\usepackage[normalem]{ulem}
\newcommand\redout{\bgroup\markoverwith{\textcolor{red}{\rule[0.5ex]{2pt}{0.8pt}}}\ULon}
\newcommand\blueout{\bgroup\markoverwith{\textcolor{blue}{\rule[0.5ex]{2pt}{0.8pt}}}\ULon}
\definecolor{since}{rgb}{0.5,0.5,0.5}

\definecolor{neworange}{HTML}{c98702}

\begin{document}
\title{R{\'e}nyi divergence-based uniformity guarantees for $k$-Universal Hash Functions\thanks{
Madhura Pathegama was with the Dept. of ECE and ISR, University of Maryland, College Park, MD 20742. He is now with the School of ECE at Georgia Institute of Technology, Atlanta, GA 30332. Email: pankajap@umd.edu. His research was supported in part by NSF grants CCF2104489 and CCF2330909.

Alexander Barg is with the Dept. of ECE and ISR, University of Maryland, College Park, MD 20742. Email: abarg@umd.edu. His research was supported in part by NSF grants CCF2110113 (NSF-BSF), CCF2104489, and CCF2330909.
}}

	\author{Madhura Pathegama,~\IEEEmembership{Graduate Student Member,~IEEE,}
		and Alexander~Barg,~\IEEEmembership{Fellow,~IEEE}}
        
	\date{}

\maketitle

\begin{abstract}
Universal hash functions map the output of a source to random strings over a finite alphabet, aiming to approximate the uniform distribution on the set of strings. A classic result on these functions, called the Leftover Hash Lemma, gives an estimate of the distance from uniformity based on the assumptions about the min-entropy of the source. We prove several results concerning extensions of this lemma to a class of functions that are $k^\ast$-universal, i.e., $l$-universal for all $2\le l\le k$. As a common distinctive feature, our results provide estimates of closeness to uniformity in terms of the $\alpha$-R{\'e}nyi divergence for all $\alpha\in (1,\infty]$. For $1\le \alpha\le k$ we show that it is possible to convert all the randomness of the source measured in $\alpha$-\Renyi entropy into approximately uniform bits with nearly the same amount of randomness. For large enough $k$ we show that it is possible to distill random bits that are nearly uniform, as measured by min-entropy. We also extend these results to hashing with side information.
\end{abstract}


 \thispagestyle{empty}

\section{Introduction}

Uniform random bit-strings are a fundamental resource in both computer science and cryptography. 
In computer science, many algorithms leverage randomization to solve problems more efficiently \cite{motwani1995randomized}.
Moreover, uniform random bits are indispensable in many cryptographic applications such as randomized encryption schemes \cite{rivest1983randomized}, secret sharing \cite{shamir1979share} , bit commitment \cite{blum1983coin}, and zero-knowledge proofs \cite{goldwasser1989knowledge}. To obtain a uniform distribution from a random source with low entropy, one 
attempts to convert its randomness into uniform bits. The maximum uniform bits extractable from a random source with is called the intrinsic randomness \cite{vembu1995generating}, and if the distribution of the source is known, a deterministic function can transform most of the entropy into uniform $q$-ary symbols.

A more common scenario in cryptographic applications is when the source distribution is unknown and cannot be efficiently estimated. In such cases, one instead relies on aggregate quantitative measures of source's randomness such as min-entropy or collision entropy. 
In computer science and cryptography, randomized mappings that send the output of the source to binary strings with small statistical distance from uniform strings, are known as {\em randomness extractors} \cite{nisan1996extracting}. If in addition to converting the randomness of an unknown source into nearly uniform bits, the function's output remains almost independent of its internal randomness, it is referred to as a strong extractor. 

A class of good extractors arises from universal hash function families \cite{carter1977universal}; see also \cite{motwani1995randomized, tyagi2023information}. A key result in this context is the leftover hash lemma (LHL) \cite{impagliazzo1989pseudo}, which shows that universal hash functions can convert a source’s min-entropy into an almost uniform bit string, with the deviation measured by total variation distance. The use of universal hash functions as strong extractors is extensively discussed in \cite{impagliazzo1989recycle}. In certain cryptographic applications, legitimate parties are required to distill uniform bits in the presence of an adversary, ensuring the adversary’s information is independent of the distilled bits. This process, called {\em privacy amplification}, can be accomplished relying on a strengthened version of the LHL proved in \cite{bennett1995generalized}, and it has gained prominence in information-theoretic security. In particular, it underpins the security of many cryptographic primitives, including secure key generation in both classical \cite{bennett1995generalized} and quantum cryptography \cite{Renner2008}, secrecy in wiretap channels \cite{tyagi2014explicit}, signature schemes \cite{amiri2018efficient}, authentication \cite{stinson1994universal}, and oblivious transfer \cite{crepeau2006optimal}. 
In essence, universal hash functions form a vital tool in information-theoretic security, particularly for quantifying the
feasibility ranges of the aforementioned protocols \cite{tyagi2023information}.

Further developments on LHL relaxed its original reliance on min-entropy ($\infty$-\Renyi entropy) to collision entropy ($2$-\Renyi entropy) \cite{bennett1995generalized}, and later to the $\alpha$-\Renyi entropy $\alpha \in (1,2]$ \cite{hayashi2011exponential}. 
Another related refinement replaced min-entropy with smoothed min-entropy \cite{Renner2008}. 
Extensions of LHL-like results have been achieved for different variations of hash functions. 
For example, \cite{tsurumaru2013dual} introduced a version of LHL for $\epsilon$-almost dual universal hash functions. 
Uniformity guarantees for linear hash functions were provided in \cite{alon1999linear, dhar2022linear} and more
recently in \cite{pathegama2024r}.

Early studies of universal hash functions relied on measuring uniformity of the distilled bits using the total variation distance or KL divergence. At the same time, some applications call for stronger measures of uniformity.
For instance, in random number generation, new standardization proposals recommend min-entropy as a metric for randomness \cite{killmann2011proposal, turan2018recommendation}. Moreover, if the adversaries are assumed to have no limits of computing
power, as happens, for instance, in information-theoretic cryptography, secrecy bounds based on total variation distance may be inadequate. Such an adversary could exploit small deviations from uniformity by collecting a large (potentially exponential) number of samples, leading to effective attacks. To counter such attacks, researchers have resorted to more stringent secrecy measures. In particular, the guessing secrecy concept of \cite{alimomeni2012guessing} assumed that secrecy is measured using min entropy. Building on this idea, \cite{iwamoto2014information} proposed a more general security framework based on $\alpha$-\Renyi divergence. These concepts were extended to lattice-based cryptography in \cite{bai2018improved}. 

\vspace*{.1in}
\emph{Our results.} Motivated by these works, in this paper we prove a version of the LHL that relies on higher-order \Renyi divergences,
offering stronger uniformity guarantees. We are not the first to report results of this kind. 
For instance, the authors of \cite{hayashi2016equivocations} derived uniformity guarantees based on $\alpha$-\Renyi divergence for $\alpha \in [0,2]$,
using $2$-universal hash functions. 
Generalizing these results, we derive uniformity guarantees for the entire range of order values $\alpha\in [2,\infty]$ in the form of one-shot estimates.

To obtain stronger uniformity guarantees for arbitrary sources, we study a class of hash functions which we call $k^*$-universal $(k \geq 2)$. A hash family is called $k^*$-universal if it is $l$-universal \cite{carter1977universal} for every $l \in \{2, \dots, k\}$, meaning that for any $l$-tuple of distinct inputs, the collision probabilities are low. 
The most common case is $k = 2$, where $2$-universality and $2^*$-universality are equivalent. 
In this case,  the corresponding mappings are called simply {\em universal hash functions}, omitting the reference to $k$. We will follow this convention in our paper.

As our main result (Theorems~\ref{thm: main 0},~\ref{thm: main2}), we show that using $k^*$-universal hash functions, it is possible to extract nearly $H_{\alpha}(X)$ random bits from the source $X$ where the distilled output is required to be approximately uniform in terms of 
the $\alpha$-\Renyi divergence with $\alpha \in (1, k]$. 
Additionally, we obtain uniformity guarantees based on conditional $\alpha$-\Renyi divergence for the case  $\alpha > k$, which reduces the dependence between the hash values and the random seed.
When $\alpha = \infty,$ this provides explicit bounds for approximating uniformity under the 
conditional $\infty$-\Renyi divergence, offering strong guarantees for uniformity in cryptographic applications.

Finally, we extend our version of the LHL lemma to account for side information. 
Suppose we aim to convert a weak source $X$ into a nearly uniform distribution, while the adversary has access to a correlated random variable $Z$. Our goal is to distill uniform random bits that are almost independent of $Z$, 
accomplishing the privacy amplification task. We show that even in this case, it is possible to provide strong uniformity and independence guarantees. The proofs in this case follow the unconditional LHL and other theorems, replacing
the \Renyi entropy by its conditional version $H_{\alpha}(X|Z)$ as the randomness measure.

The problem close to the ones that we study in this paper was previously considered in \cite{kavian2023output}. There are two main differences between that paper and our work. First, \cite{kavian2023output} considers the set of all 
possible hash functions, whereas we consider a much smaller set of $k$-universal hash functions. 
Second, unlike our work, the authors of \cite{kavian2023output} focus on the asymptotic setting for the bounds on $\alpha$-divergence. Additional remarks in the main text provide further details about the similarities and differences.

\section{Preliminaries}

We begin by establishing the notation used throughout the paper. 
Let $q \geq 2$ be an integer, and let $\zz_q^m$ be the set of length-$m$ strings over the alphabet $\{0,1,\dots,q-1\}$. 
For a finitely supported random variable $Z$, we denote its probability mass function by $P_Z$. 
If $Z$ follows a probability distribution $P$, we write $Z \sim P$ to indicate that $P_Z = P$. 
When $Z$ is uniformly distributed over a set $\cA$, we write $Z \sim \cA$ with some abuse of notation.
Denote by $U_m$ the uniform random variable on $\zz_q^m$ and let $P_{U_m}$ denote its distribution. Unless stated otherwise, all random variables in this work are assumed to be defined on finite spaces.

\subsection{Measures of randomness} 
We employ \Renyi entropies to quantify the randomness of random variables. For $\alpha \in (1, \infty)$, the \Renyi entropy of a random variable $X$ is defined as:
\begin{align}
   H_{\alpha} (X)&=\frac1{1-\alpha}\log_q\Big(\sum_x P_X(x)^{\alpha} \Big), \label{eq:Renyi}
 \end{align}  
with the limiting cases $\alpha = 1, \infty$ given by
   \begin{align*}
		H_1(X)&=-\sum_x P_X(x)\log_q {P_X(x)}\\
           H_{\infty}(X)&= \min_x (-\log_q {P_X(x)}).  
\end{align*}
Of course, for $\alpha=1$ the \Renyi entropy coincides with the Shannon entropy, which we denote simply as $H(X)$. The quantity  
$H_{\infty}(X)$ is commonly referred to as the min-entropy, while the common term for $ H_2(X)$ is the collision-entropy.
We observe that $H_{\alpha}(X)$ decreases as $\alpha$ increases, while $\frac{\alpha - 1}{\alpha} H_\alpha(X)$ increases with $\alpha$. These relationships enable us to bound \Renyi entropies of different orders in terms of one another. \Renyi entropies can be also defined for $0<\alpha<1$, though we do not consider this range in our work.

Since our random variables take values in $q$-ary product spaces, we use base-$q$ logarithms throughout (any other base could
be used instead as long as it is consistent throughout the paper).

\subsection{Proximity measures for distributions}
There are several ways to measure proximity between two probability distributions. A commonly used metric is the  {\em total variation distance} $d_{\text{TV}}(\cdot,\cdot)$ which is defined as follows: Let $P$ and $Q$ be two discrete probability measures defined on the space $\cX$. Then 
\begin{align*}
    d_{\text{TV}}(P,Q) = \max_{A\subset \cX}|P(A)-Q(A)|.
\end{align*}
Another common metric is the KL divergence, given by
\begin{align*}
    D(P\|Q)&=\sum_x P(x)\log_q \frac{P(x)}{Q(x)}.
\end{align*}
Traditionally, total variation distance and KL-divergence have been used to assess how close a distribution is to uniform. 
In this work, we adopt a stricter measure, namely, the \Renyi divergence of order $\alpha > 1$.
For two discrete distributions $P$ and $Q (P \ll Q)$, defined on the same probability space $\cX$, and for $\alpha \in (1, \infty)$, the \Renyi divergence is defined as
\begin{align}
    D_{\alpha}(P\|Q)=\frac{1}{\alpha -1} \log_q \sum_x P(x)^{\alpha} Q(x)^{-(\alpha -1)}.
\end{align}
Taking limits we obtain,
   \begin{align*}
		D_1(P\|Q)&= D(P\|Q)\\[.1in]
          D_\infty(P\|Q)&=  \max_x \log_q \frac{P(x)}{Q(x)} .
  \end{align*}
For simplicity, we say $\alpha$-divergence instead of the \Renyi divergence of order $\alpha$.
Note that $D_\alpha$ is monotone increasing, i.e., for $1\le \alpha<\alpha'$ we have $D_\alpha(P\|Q)\le D_{\alpha'}(P\|Q)$. 
Therefore, higher $\alpha$-divergences provide stronger bounds for proximity between two distributions. Also note that if $Q$ is uniform, then $D_\alpha(P\|Q)=\log_q|\cX|-H_{\alpha}(P)$.

Yet another proximity measure between probability vectors is the $l_\alpha$ distance $\|P-Q\|_{l_\alpha}$, but
for $\alpha>1$ it is essentially equivalent to $D_\alpha$ \cite{pathegama2024r} (and $d_{l_1}=\frac 12 d_{\text{TV}}$).
For this reason we will not mention it below. 

\subsection{Hash functions}\label{sec: hash}

\begin{definition}\label{def:h0}
Let $\cX$ be a finite set.  A family of hash functions $H = \{h: \cX \to \zz_q^m\}$ is $k$-universal if for any distinct elements $(x_1, \dots, x_k) \in \cX^k$, we have 
    \begin{align*}
        \Pr_{h\sim H}(h(x_1)=h(x_2) =  \dots = h(x_k)) \leq q^{-m(k-1)}.
    \end{align*}
\end{definition}

The random selection of $H$ in the above definition can be modeled as a uniform random variable $S$ over a set $\cS$ which is called the {\em seed}. Adopting this point of view, the family $H$ can be viewed as a single function $h$ defined on $\cS \times \cX$. In such cases, we refer to $h$ as a hash function (as opposed to a single realization of the hash family). Accordingly, we can rewrite Def.~\ref{def:h0} as follows.

\begin{definition}\label{def:kUHF}
    We call a function $ h: \cS \times \cX \to \zz_q^m$ $k$-universal if for any distinct $(x_1, \dots, x_k) \in \cX^k$,
    \begin{align}\label{eq: k-hash}
        \Pr_{S \sim \cS}(h(S , x_1)= h(S , x_2)=  \dots = h(S , x_k)) \leq q^{-m(k-1)}.
    \end{align}
For $k=2$ we call $h$ a universal hash function, omitting the mention of $k$.    
\end{definition}
In this work, we rely on this definition of $k$-universal hash functions, thinking of $h$ as a single deterministic function on $\cS \times \cX$.
We now introduce a somewhat non-standard notion of universality, which will be used to present most of our results.

\begin{definition}\label{def:k*HF}
    We call a function $ h: \cS \times \cX \to \zz_q^m$ $k^*$-universal if it is $l$-universal for all $l \in \{2,3,\dots,k\}$.
\end{definition}
\remove{A notable subclass of $k^*$-universal hash functions is the $k$-wise independent hash functions, defined below.

\begin{definition}\label{def:kIHF}
    We call a function $ h: \cS \times \cX \to \zz_q^m$ $k$-wise independent if for any distinct $(x_1, \dots, x_k) \in \cX^k$ and any $(y_1,\dots,y_k) \in (\zz_q^m)^k$ (not necessarily distinct),
    \begin{align}\label{eq: k-ihash}
        \Pr_{S \sim \cS}(h(S , x_1)=y_1, \dots, h(S , x_k)=y_k) = q^{-mk}.
    \end{align}    
\end{definition}}
As noted in the introduction, universal hash functions are commonly employed in distilling uniform bits. This property has been applied in numerous proofs related to information-theoretic secrecy \cite{tyagi2023information}. The process of distilling uniformity using universal hash functions is formalized in the well-known LHL lemma, which we state below.

\begin{proposition}[Leftover hash lemma \cite{impagliazzo1989pseudo}]
    Let $X$ be a random variable defined on $\cX$ and let $S$ be a uniform random variable $S \sim \cS$ that is independent of $X$. Let $h : \cS \times \cX \to \zz_q^m$ be a universal hash function.  If $m \leq H_{\infty}(X)-\log_q(1/\epsilon)$, then 
    \begin{align}\label{eq: impag}
         d_{\text{TV}}(P_{h(S;X),S},P_{U_m}P_S) \leq \frac{\sqrt{\epsilon}}{2}.
    \end{align}
\end{proposition}

Many variations and improvements of the above statement appeared later \cite{barak2011leftover,fehr2014conditional}. For instance, the authors of \cite{bennett1995generalized} showed that it is possible to replace the requirement on $m$ with a less restrictive one: $m \leq H_{2}(X)-\log_q(1/\epsilon)$, which yields 
\begin{align}\label{eq: bennet}
    D(P_{h(S;X),S}\|P_{U_m}P_S) & \leq \frac{\epsilon}{\ln q}.
\end{align}
Even with this revised condition, the bound for total variation distance in \eqref{eq: impag} remains valid.

A further improvement based on the \Renyi entropy was provided in \cite{hayashi2011exponential}. We state this result below, adapted to our notation.

\begin{proposition}\cite{hayashi2011exponential,hayashi2016equivocations}\label{prop: Hayashi}
    Let $X$ be a random variable defined on $\cX$ and let $S$ be a uniform random variable that is independent of $X$. Let $h : \cS \times \cX \to \zz_q^m$ be a universal hash function and let $\alpha \in (1,2]$ If $m \leq H_{\alpha}(X)-\frac{1}{\alpha-1}\log_q (\frac{1}{\epsilon (\alpha-1) \ln q })$, then for $S \sim \cS$ and independent of $X$,
  \begin{align}
       D_\alpha(P_{h(S;X),S}\|P_{U_m}P_S) \leq \epsilon.
    \end{align}
\end{proposition}
\begin{remark}
Addressing the cryptographic context, papers \cite{bennett1995generalized}, \cite{hayashi2011exponential} (see also \cite[p.~126]{tyagi2023information}) state a version of LHL that accounts for side information available to the adversary in the form of a random variable $Z$
correlated with the source $X$. For our main results, we remove this assumption, which simplifies the presentation. At the same time, it can be easily added to the statements and proofs, as we show later in Section \ref{sec: side}.
\end{remark}

 Another variation of the LHL is based on the $\epsilon$-smoothed $\alpha$-R{\'e}nyi entropy which is defined below. 

 \begin{definition}
     Let $X$ be a random variable on $\cX$. The $\epsilon$-smoothed $\alpha$-R{\'e}nyi entropy \cite{renner2004smooth} of $X$ is defined as 
     \begin{align*}
        H_{\alpha}^\eta(X) = \min_{P_Y \in \cB_\eta(P_X)}H_\alpha(Y),
     \end{align*}
     where $\cB_\eta(P_X) = \{P: d_{TV}(P, P_X) \leq \eta \}$ is a TV ball of radius $\eta$ around $P_X$.
 \end{definition}

 \begin{proposition}[\cite{Renner2008}{ Corollary 5.6.1}]\label{thm: smoothed}
     Let $X$ be a random variable defined on $\cX$ and let $S\sim \cS$ be independent of $X$. Let $h : \cS \times \cX \to \zz_q^m$ be a universal hash function. If $m \leq H^\eta_{\infty}(X)-\log(1/\epsilon)$, then 
     \begin{align}
         d_{TV}(P_{h(S;X),S},P_{U_m}P_S) \leq 2\eta + \frac{\sqrt{\epsilon}}{2}.
     \end{align}
 \end{proposition}

In summary, the earlier results used the total variation distance or low-order \Renyi divergences 
to measure the uniformity of the hashed source output. In this work we extend this suite of results
to rely on stronger uniformity measures $D_\alpha, \alpha\ge 2.$

\section{\texorpdfstring{$k$}--universality and uniformity guarantees}
In this section, we establish uniformity guarantees for $k^*$-universal hash functions. We prove that for any source $X$ and a $k^*$-universal hash function $h$, it is possible to extract almost $H_\alpha(X)$ random bits for any $\alpha \in (1, k]$. 
The proof comprises two stages, of which the first handles the case of integer $\alpha$'s and the second ``fills the gaps''.
We begin with the integer case, which also allows us to state the result in a compact form.

\begin{theorem}\label{thm: main 0}
    Let $\epsilon >0 $ and let $k \in \{2,3\dots\}$ and $\alpha \in \{2,3,\dots,k\}$. Let $X$ be a random variable defined on $\cX$ and let $S\sim \cS$ be a random variable independent of $X$. If 
     $$
     m \leq H_{\alpha}(X) -  \log_q\Big(\frac{\alpha^2}{2\epsilon(\alpha-1)\ln q}\Big),
     $$
     then for a $k^*$-universal hash function $h : \cS \times \cX \to \zz_q^m$,
    \begin{align*}
        D_\alpha(P_{h(S;X),S}\|P_{U_m}P_S) \leq \epsilon.
    \end{align*}
\end{theorem}
This theorem is a slight relaxation of Theorem \ref{thm: main1} which we present below. Starting with this version enables us to align the statement with the format of classic LHL statements in Section \ref{sec: hash}, and also uses a compact form of the inequality
for $D_\alpha$. 

Before stating the theorem, we remind the reader the definition fo Stirling numbers. The Stirling number of the second kind, denoted $\stirlingII kl$, equals the number of ways to partition a $k$-set into $l$ parts (see, e.g., \cite[Ch.~5]{comtet2012advanced}). Stirling numbers can also be
defined via their generating function
    $$
    z^k=\sum_{l=1}^k \stirlingII kl z(z-1)\dots(z-l+1).
    $$

\begin{theorem}\label{thm: main1}
     Let $k \in \{2,3\dots\}$. Let $X$ be a random variable defined on $\cX$ and let $S\sim \cS$ be a random variable independent of $X$. Let $h : \cS \times \cX \to \zz_q^m$ be $k^*$-universal. Then
    \begin{align}\label{eq: main1}
        q^{(k-1)D_k(P_{h(S;X),S}\|P_{U_m}P_S)} \leq \sum_{l=1}^{k}\stirlingII{k}{l}q^{(k-l)(m-H_{k}(X))}.
    \end{align}
    Moreover, \eqref{eq: main1} also holds if all instances of $k$ in it are replaced with any integer $\alpha$
    between 2 and $k$.
\end{theorem}

Before presenting the proof, we need to introduce some notation. We abbreviate a $k$-tuple $x_1, \dots, x_k \in \cX$ as $x^k.$ In the proof, we sum a particular quantity over all $x^k \in \cX^k$. 
To simplify this summation, we first partition $\cX^k$ into specific blocks and then split the sum into a double sum, first within each block and then across the blocks. 
To construct this partition $T$ on $\cX^k$, we will use partitions of an auxiliary set, $\{1, 2, \dots, k\}$, which also give rise to the Stirling numbers in the final answer.

Let $\mathfrak{P}_k$ and $\mathfrak P_k(l)$ denote the set of all partitions and the set of all $l$-partitions (partitions into $l$ blocks) of $\{1, 2, \dots, k\}$, respectively. We write an $l$-partition $\cP \in \mathfrak P_k(l)$ as $\cP = \{\cP_1,\dots,\cP_l\}$. 
Next, we construct a partition $T$ of $\cX^k$, where the blocks are indexed by the elements of $\mathfrak{P}$, i.e., $T = \{T_\cP\}_{\cP \in \mathfrak{P}_k}$.
The rule for assigning elements of $\cX^k$ to the blocks of $T$ is as follows: an element $x^k \in T_\cP$ if and only if, for all $i, j \in \{1,2,\dots,k\}$ within the same block of $\cP$, we have $x_i = x_j$, and for any $i, j$ in different blocks, $x_i \neq x_j$.

As a simple example to clarify our notation, let $\cX = \{0,1\}$ and $k=2$. There are 2 different partitions of the $2$-set, namely:
\begin{align*}
    \mathfrak{P}_2 &= \{\; \{\{1,2\}\},\{\{1\},\{2\}\}\;\}. 
\end{align*}
The corresponding blocks of $T$ are
\begin{align*}
        T_{\{\{1,2\}\}} = \{(0,0), (1,1)\} , \quad 
    T_{\{\{1\},\{2\}\}} = \{(0,1), (1,0)\}. 
\end{align*}
With this notation in place, we now proceed to the proof.
\begin{proof} (of Theorem \ref{thm: main1})
The expression on the left in \eqref{eq: main1} is simply the expectation
\begin{equation}\label{eq: DE}
q^{(k-1)D_k(P_{h(S;X),S}\|P_{U_m}P_S)}=E_{P_{U_m}P_S} \Big[\frac{P_{h(S;X),S}(\cdot,\cdot)}{P_{U_m}(\cdot)P_S(\cdot)}\Big]^{k-1}.
\end{equation}
Accordingly, we compute
\begin{align}
        & q^{(k-1)D_k(P_{h(S;X),S}\|P_{U_m}P_S)}  =  \sum_{u\in \zz_q^m,s\in\cS}\frac{P_{h(S;X),S}(u,s)^{k}}{P_{U_m}(u)^{k-1}P_S(s)^{k-1}}\nonumber\\
               &\;= q^{m(k-1)}\sum_{s}P_S(s)\sum_{u}P_{h(S;X)|S}(u|s)^{k}\nonumber\\
        &\;= q^{m(k-1)}\sum_{s}P_S(s)\nonumber\\
        &\hspace*{.5in}\times\sum_{u}\prod_{i=1}^k\Big[\sum_{x_i \in \cX}P_{h(S;X)|S,X}(u|s,x_i)P_X(x_i)\Big] \nonumber\\ 
        &\;= q^{m(k-1)}\sum_{s}P_S(s)\nonumber\\
        &\hspace*{.5in}\times\sum_{u}\prod_{i=1}^k\Big[\sum_{x_i \in \cX}\1\{h(s,x_i) = u\}P_X(x_i)\Big].
        \label{eq: u0}
\end{align}
For typographical purposes below we write $P_{X^k}(x^k):=P_X(x_1)\dots P_X(x_{k})$. Continuing from \eqref{eq: u0}
      \begin{align}
        &\;= q^{m(k-1)} \sum_{u} \sum_{x^k \in \cX^k} {P_{X^k}(x^k)} \nonumber\\
        &\hspace*{.5in}\times \sum_{s}P_S(s) 
        \prod_{i=1}^k \1\{h(s,x_i)=u\}
        \remove{\1\{h(s,x_1) = u\}  \dots    \1\{h(s,x_k) = u\}}\label{eq: u1}\\
        &\;= q^{m(k-1)} \sum_{u} \sum_{x^k \in \cX^k} P_{X^k}(x^k)\nonumber\\
        &\hspace*{.5in}\times \Pr_S\big(h(S,x_1) = \dots 
        = h(S,x_k) = u \big) \nonumber\\
        &\;= q^{m(k-1)} \sum_{x^k \in \cX^k} P_{X^k}(x^k) \Pr_S\big(h(S,x_1) = \dots 
        = h(S,x_k) \big) \label{eq: u11}.
    \end{align}
Let $\eta(x^k)$ be the number of distinct entries in $x^k$. From $k^*$-universality, we have
\begin{align*}
    \Pr_S\big(h(S,x_1) =  \dots = h(S,x_k) \big) \leq 
    q^{-m((\eta(x^k)-1)}.
\end{align*}
Continuing the calculation,
\begin{align}
    &q^{(k-1)D_k(P_{h(S;X),S}\|P_{U_m}P_S)}\nonumber\\
    &\;\leq q^{m(k-1)} \sum_{x^k \in \cX^k} P_{X^k}(x^k) q^{-m (\eta(x^k)-1)} \nonumber\\
    &\;= q^{m(k-1)} \sum_{l = 1} ^k \sum_{\cP \in \mathfrak P_k(l)}\sum_{x^k \in T_\cP} P_{X^k}(x^k) q^{-m (\eta(x^k)-1)} \nonumber\\
    &\;= q^{m(k-1)} \sum_{l = 1} ^k q^{-m (l-1)} \sum_{\cP \in \mathfrak P_k(l)}\sum_{x^k \in T_\cP} P_{X^k}(x^k)  \label{eq: u2}.
\end{align}
Note the transition form the number of distinct entries to partitions into $l$ parts in \eqref{eq: u2}.
Let us fix an $l$-partition $\cP = \{\cP_1,\dots,\cP_l\}$ and set $p_i := |\cP_i|$. Denote a generic element of
the block $\cP_i$ by $\pi_i$. We now estimate the innermost sum in \eqref{eq: u2}:
\begin{align}\label{eq: u22}
    \sum_{x^k \in T_\cP} P_{X^k}(x^k)
     & = \sum_{x^k \in T_\cP} \prod_{j=1}^l P_X(x_{\pi_j})^{p_j}  
    \leq \prod_{j=1}^l \sum_{x \in \cX} P_X(x)^{p_j},
\end{align}
where the inequality is obtained by removing the requirement that the variables in different blocks must be distinct. 

Let us fix $j$ and evaluate the sum on $x$ above. If $p_j = 1$, evidently, the sum is equal to $1$. In particular, if  $ l=k,$  i.e. $p_j = 1$ for all $j$, we have 
\begin{align}\label{eq: k=l}
    \sum_{x^k \in T_\cP} P_{X^k}(x^k) \leq 1.
\end{align}
If $p_j > 1$, we may write
\begin{align*}
    \sum_{x \in \cX}P_X(x)^{p_j} = \sum_{x \in \cX}\big(P_X(x)^{\frac{k-l+1-p_j}{k-l}}\big)
      \big(P_X(x)^{k-l+1}\big)^{\frac{p_j-1}{k-l}}.
\end{align*}
Since $p_j >1$, $l<k$, so all the quantities on the right-hand side are well defined.
Now let us use H\"older's inequality $\|fg\|_1\le \|f\|_\lambda\|g\|_\mu$ with $f$ and $g$ given by the terms in 
the parentheses and with the exponents $\lambda=\frac{k-l}{k-l+1-p_j}$ and $\mu=\frac{k-l}{p_j-1}$. We obtain
\begin{align}\label{eq: Holder}
    \sum_{x \in \cX}P_X(x)^{p_j} &{\leq} \big(\sum_{x \in \cX}P_X(x)\big)^{\frac{k-l+1-p_j}{k-l}} \big(\!\sum_{x \in \cX}P_X(x)^{k-l+1}\big)^{\frac{p_j-1}{k-l}}\nonumber\\
    & = \big(\sum_{x \in \cX}P_X(x)^{k-l+1}\big)^{\frac{p_j-1}{k-l}}.
\end{align}
Therefore, if $l <k$, then irrespective of the value of $p_j$ we have the estimate
\begin{align}\label{eq: Holder2}
    \sum_{x \in \cX}P_X(x)^{p_j} \leq \big(\sum_{x \in \cX}P_X(x)^{k-l+1}\big)^{\frac{p_j-1}{k-l}}.
\end{align}
Returning to \eqref{eq: u22}, for all $l<k$ we obtain 
\begin{align}
    \sum_{x^k \in T_\cP} P_{X^k}(x^k) 
    & \leq \prod_{j=1}^l \sum_{x \in \cX} P_X(x)^{p_i} \nonumber \\
    & \leq \prod_{j=1}^l \Big[\sum_{x \in \cX}P_X(x)^{k-l+1}\Big]^{\frac{p_j-1}{k-l}} \nonumber \\
    &= \sum_{x \in \cX}P_X(x)^{k-l+1} \nonumber \\
    &= q^{-(k -l)H_{k-l+1}(X)}. \label{eq: prob bound}
\end{align}
From \eqref{eq: k=l} it can be easily seen that the inequality $\sum_{x^k \in T_\cP} P_{X^k}(x^k)\le q^{-(k -l)H_{k-l+1}(X)}$ holds also for $k=l$. 
Now let us substitute these results into \eqref{eq: u2}. Recalling that $\stirlingII{k}{l}$ counts the number of partitions into
$l$ blocks, we can write
    \begin{align*}
        & q^{(k-1)D_k(P_{h(S;X),S}\|P_{U_m}P_S)} \\
        & \qquad \leq q^{m(k-1)}\sum_{l=1}^k\stirlingII{k}{l} q^{-m(l-1)}q^{-(k-l)H_{k-l+1}(X)}\\
        & \qquad = \sum_{l=1}^k\stirlingII{k}{l} q^{(k-l)(m-H_{k-l+1}(X))}\\
        & \qquad \leq \sum_{l=1}^k\stirlingII{k}{l} q^{(k-l)(m-H_{k}(X))},
    \end{align*}
where the final inequality holds due to the monotonicity property of \Renyi entropy.
    
The final claim follows from the fact that if $h$ is $k^*$-universal, it is also $\alpha^*$-universal for the integer
$\alpha$ between $2$ and $k$.
\end{proof}
\begin{remark} 
We note that asymptotic estimates of the quantity measuring closeness to uniformity in terms of
\Renyi entropy of various orders $\alpha\in (1,\infty]$ were previously considered in \cite{kavian2023output}. In particular, its authors also start with an expression equivalent to our Eq.~\eqref{eq: DE}. At the same time, the proof method in \cite{kavian2023output} utilizes the independence of hash inputs, whereas we adopt a more general setting of $k$-universal hash functions. Further, the main results in \cite{kavian2023output} are concerned with establishing the asymptotic behavior for the expectation on the right-hand side of \eqref{eq: DE}; see Theorems 1 and 4 in that work, whose authors did not pursue finite-length bounds of the form obtained in this paper. 

\end{remark}

Observe that the right-hand side of \eqref{eq: main1} corresponds to the $k$-th moment of a Poisson random variable,  normalized by its mean.  Recall that the $k$-th moment of a Poisson random variable $Z$ with parameter $\lambda$ is given by 
\begin{align*}
    \ee[Z^k] = \sum_{l=1}^k\stirlingII{k}{l}\lambda^l.
\end{align*}
By setting $\lambda = q^{H_k(X)-m}$, we observe that the right-hand side of \eqref{eq: main1} simplifies to $\ee[(Z/\lambda)^k]$
for $Z\sim\text{Poi}(\lambda)$. This insight allows us to leverage standard bounds on Poisson moments, to derive simpler bounds for $q^{(k-1)D_k(P_{h(S,X),S} \| P_{U_m}P_S)}$. A simple and well-known bound for Poisson moments is as follows: 
\begin{align}\label{eq: pms}
    \ee[(Z/\lambda)^k] \leq \exp\Big(\frac{k^2}{2\lambda}\Big).
\end{align}

Using this bound, let us prove Theorem \ref{thm: main 0}
\begin{proof} (of Theorem \ref{thm: main 0})
Evidently,
  \begin{align}\label{eq: q-D_k bound 2}
    q^{(k-1)D_k(P_{h(S;X),S}\|P_{U_m}P_S)} \leq \exp\Big(\frac{k^2}{2q^{H_k(X)-m}}\Big),
\end{align}
which implies 
\begin{align}\label{eq: D_k bound 2}
    D_k(P_{h(S;X),S}\|P_{U_m}P_S) \leq \frac{k^2}{2 q^{H_k(X)-m}(k-1)\ln q}.
\end{align}

If we set $m \leq H_k(X) - \log_q\Big(\frac{k^2}{2\epsilon(k-1)\ln q}\Big)$, then we have $D_k(P_{h(S;X),S}\|P_{U_m}P_S) \leq \epsilon$. This addresses the case $\alpha = k.$ Using the last claim of Theorem~\ref{thm: main1}, we can extend this argument to apply to all $\alpha \in \{2,\dots,k\}$. 
\end{proof}

\begin{remark}\label{rmk: gamma}
Of course, \eqref{eq: pms} is not the best possible estimate of the moments, and tighter results are available.
For instance, using Theorem 1 of \cite{ahle2022sharp}, we obtain the bound
\begin{multline}\label{eq: D_k bound}
    D_k(P_{h(S;X),S}\|P_{U_m}P_S) \\\leq \frac{k}{k-1}\log_q\Big(\frac{kq^{m-H_k(X)}}{\ln(kq^{m-H_k(X)}+1)}\Big).
\end{multline}
With this we can strengthen Theorem \ref{thm: main 0}, claiming that its conclusion holds under a more forgiving assumption: $m \leq H_{k}(X) + \log_q\Big(\frac{\gamma(q^{\epsilon^\frac{k}{k-1}})}{k}\Big)$, where $\gamma(y)$ is the unique solution $x$ to the equation $\frac{x}{\ln(x+1)} = y, y \geq 1$. Since $m$ is now allowed to take larger values, this supports extracting more nearly uniform bits from the source, which accounts for the stronger outcome.
\end{remark}

Our next task is to move from integer $\alpha$'s to all real values $1<\alpha\le k$, generalizing Theorem \ref{thm: main1}.

\begin{theorem}\label{thm: main2}
     Let $k \in \{2,3\dots\}$ and $\alpha \in (1,k]$. Let $X$ be a random variable defined on $\cX$ and let $S$ be a uniform random variable that is independent of $X$. Let $h : \cS \times \cX \to \zz_q^m$ be $k^*$-universal. Then
    \begin{align}\label{eq:continuous}
        & q^{(\alpha-1)D_\alpha(P_{h(S;X),S}\|P_{U_m}P_S)}\nonumber\\ 
        & \qquad \leq  \sum_{l=1}^{\lceil \alpha \rceil-1} l \stirlingII{\lceil \alpha \rceil-1}{l} q^{(\alpha-l)(m-H_{\alpha}(X)} \nonumber\\
        & \qquad \quad \quad + \sum_{l=1}^{\lceil \alpha \rceil}\stirlingII{\lceil \alpha\rceil-1}{l-1} q^{(\lceil \alpha \rceil-l)(m-H_{ \alpha }(X)}. 
    \end{align}
\end{theorem}
The proof of this theorem is given in Appendix \ref{Appendix A}.

Note that when $\alpha = k$, Theorem~\ref{thm: main2} recovers Theorem \ref{thm: main1} due to the identity $\stirlingII{k}{l}=l \stirlingII{k-1}{l}+\stirlingII{k-1}{l-1}$.

\begin{remark}
  In a number of cases it is possible to further simplify the right-hand side of \eqref{eq:continuous}. For instance, if $m \leq H_\alpha(X)$, we have $q^{m-H_{\alpha}(X)} \leq 1 $. Consequently, by replacing,  $(\lceil\alpha\rceil-l)$ with $(\alpha -l)$ in the exponent of $q$ in the last sum, we obtain:
    \begin{align}\label{eq:continuous2}
        q^{(\alpha-1)D_\alpha(P_{h(S;X),S}\|P_{U_m}P_S)} \leq \sum_{l=1}^{\lceil \alpha \rceil}\stirlingII{\lceil \alpha \rceil}{l} q^{( \alpha -l)(m-H_{\alpha}(X)}. 
    \end{align}   
    On the other hand, if $m > H_\alpha(X)$, a similar argument yields the inequality
    \begin{align}\label{eq:continuous3}
        q^{(\alpha-1)D_\alpha(P_{h(S;X),S}\|P_{U_m}P_S)} \leq \sum_{l=1}^{\lceil \alpha \rceil}\stirlingII{\lceil \alpha \rceil}{l} q^{(\lceil \alpha \rceil-l)(m-H_{\alpha}(X)}. 
    \end{align} 
We can use the moment bounds such as \eqref{eq: pms} to bring this estimate to the form similar to Theorem~\ref{thm: main 0}.
\end{remark}

Theorem \ref{thm: main2} also allows us to estimate the deviation of the distilled bits from uniformity in terms of $\alpha$-divergence, when $\alpha$ is close to $1$.

\begin{corollary}\label{cor: LHL (1,2]}
    Let $\epsilon > 0$ and $\alpha \in (1,2]$. Let $X$ be a random variable defined on $\cX$ and let $S\sim \cS$ be a uniform random variable independent of $X$. Let $h : \cS \times \cX \to \zz_q^m$ be $2$-universal.  If $m \leq H_{\alpha}(X)-\frac{1}{\alpha-1}\log_q (\frac{1}{\epsilon (\alpha-1) \ln q })$, then
    \begin{align}\label{eq: LHL (1,2]}
        D_\alpha(P_{h(S;X),S}\|P_{U_m}P_S) \leq \epsilon.
    \end{align}
\end{corollary}

\begin{proof}
Observe that for $k=2$ there is no difference between universality and $\ast$-universality. Letting $\alpha \in (1,2]$ and applying Theorem \ref{thm: main2} we obtain
    \begin{align}
        q^{(\alpha-1)D_{\alpha}(P_{h(S,X),S}\|P_{U_m}P_S)} \leq  q^{(\alpha-1)(m-H_\alpha(X))} + 1.
    \end{align}
    Therefore,
    \begin{align}
        D_\alpha(P_{h(S,X),S}&\|P_{U_m}P_S)\nonumber\\ &\leq \frac{1}{\alpha-1}\log_q(1+ q^{(\alpha-1)(m-H_\alpha(X))})\nonumber\\ 
        &\leq \frac{q^{(\alpha-1)(m-H_\alpha(X))}}{(\alpha-1)\ln q} \leq \epsilon.
    \end{align}
\end{proof}
Since the $\alpha$-divergence increases with $\alpha,$ inequality \eqref{eq: LHL (1,2]} is still valid if
$D_\alpha(\cdot \|\cdot)$ is replaced with the KL divergence. This recovers the claim of Proposition~\ref{prop: Hayashi}
implied by the results of \cite{hayashi2011exponential}, so our results generalize this work to all $\alpha\in[1,2]$.

\section{Distilling min-entropy}\label{sec: min}
As already mentioned, information-theoretic security results often rely on uniformity guarantees based in min-entropy
 \cite{turan2018recommendation}. 
In this section, we examine how effectively $k^*$-universal hash functions can meet these guarantees. Our results are
stated in terms of the {\em conditional} $\infty$-divergence rather than the more standard one $D_\infty(P_{h(S;X),S} \| P_{U_m} P_S)$. This new divergence measure, defined below, retains the same min-entropy guarantees but relaxes the stringent independence requirements between the seed $S$ and the extracted random bits. To explain the reasoning behind this shift, observe that
the unconditional version of $D_\infty$ accounts for the worst-case deviation from uniformity and for the least favorable seed. However, the $k^*$-universality condition does not account for 
the unfavorable seeds as it does not explicitly constrain
the behavior of joint distributions of $l$ variables with $l\gg k.$ Consequently, we opt for the conditional \Renyi divergence, which instead averages the worst-case scenario over {all seeds}.

This reasoning applies not just to $D_\infty$ but also to other $\alpha$-divergences once $\alpha \gg k$. 
For this reason, we give a more general definition of a {\em conditional \Renyi divergence of order $\alpha$}:
\begin{align*}
    D_\alpha(P_{h(S;X)}\|P_{U_m}|P_S) = \sum_{s\in \cS}P_S(s)D_\alpha(P_{h(S;X)|S}(\cdot| s)\|P_{U_m}).
\end{align*}
Jensen's inequality implies that $D_\alpha(P_{h(S;X)}\|P_{U_m}|P_S) \leq D_\alpha(P_{h(S;X),S}\|P_{U_m}P_S)$.  
Both conditions 
 $$D_\alpha(P_{h(S;X),S}\|P_{U_m}P_S) \leq \epsilon \text{\; and\; }D_\alpha(P_{h(S;X)}\|P_{U_m}|P_S) \leq  \epsilon$$ provide the same uniformity guarantee for $h(S,X)$, namely $H_\alpha(h(S,X)) \geq m -\epsilon $. 
However, $D_\alpha(P_{h(S;X)} \| P_{U_m} | P_S)$ does not penalize the correlation between $h(S,X)$ and $S$ as much as $D_\alpha(P_{h(S;X),S} \| P_{U_m} P_S)$ does.

Relying on the conditional $\alpha$-divergence as a uniformity guarantee, we can prove the following result about
$k^*$-universal hash functions that is applicable when $\alpha > k$.

\begin{proposition}\label{prop: alpha > k}
     Let $k \in \{2,3,\dots\}.$ Let $X$ be a random variable defined on $\cX$ and let $S\sim \cS$ be a uniform random variable that is independent of $X$. Let $h : \cS \times \cX \to \zz_q^m$ be $k^*$-universal. Then for $\alpha \in (k,\infty)$
    \begin{align}\label{eq: alpha > k}  
        & D_\alpha(P_{h(S;X)}\|P_{U_m}|P_S) \nonumber\\ 
        & \quad \leq \frac{\alpha-k}{k(\alpha-1)}m + \frac{\alpha}{\alpha-1}\log_q \Big(\frac{kq^{m-H_k(X)}}{\ln(kq^{m-H_k(X)}+1)}\Big)    
    \end{align}
    and
    \begin{align}\label{eq: alpha = infty}
         D_\infty(P_{h(S;X)}\|P_{U_m}|P_S) &\leq \frac{m}{k} + \log_q\Big(\frac{kq^{m-H_k(X)}}{\ln(kq^{m-H_k(X)}+1)}\Big).
    \end{align}
\end{proposition}
\begin{proof}

    First, let us define the following variation of the conditional \Renyi entropy for $\alpha \in (1,\infty)$:
    \begin{align*}
        \Tilde{H}_\alpha(X|Z) & = \frac{1}{1-\alpha} \sum_{z \in \cZ}P_Z(z) \log_q\Big( \sum_{x \in \cX}P_{X|Z}(x|z)^\alpha \Big).
    \end{align*}
 (we prefer not to call it conditional entropy because in the next section we use this term for a different quantity). Now observe that
        \begin{align}
            D_\alpha(P_{h(S;X)}\|P_{U_m}|P_S) &= m - \tilde{H}_\alpha(h(S;X)|S) \label{eq: DC alpha}.
        \end{align}
Since $\frac{\alpha-1}{\alpha}H_\alpha$ is an increasing function of $\alpha,$ so is 
$\frac{\alpha-1}{\alpha}\tilde{H}_\alpha(\cdot|\cdot)$.  
 Together with \eqref{eq: DC alpha} this implies that 
     \begin{align}
        \frac{k-1}{k}(m -& D_k(P_{h(S;X)}\|P_{U_m}|P_S)) \nonumber\\ 
        & \leq \frac{\alpha-1}{\alpha}(m -D_\alpha(P_{h(S;X)}\|P_{U_m}|P_S)),
     \end{align}
     for $\alpha >k$, or
     \begin{align*}
        D_\alpha(P_{h(S;X)}&\|P_{U_m}|P_S)\nonumber \\
        \leq  & \,\frac{\alpha(k-1)}{k(\alpha-1)}D_k(P_{h(S;X)}\|P_{U_m}|P_S) + \frac{\alpha-k}{k(\alpha-1)}m.
    \end{align*}
Combining this with \eqref{eq: D_k bound} yields \eqref{eq: alpha > k}.
    Letting $\alpha$ approach infinity in the last inequality, we obtain 
    \begin{align}\label{eq: infty bound}
        D_\infty(P_{h(S;X)}\|P_{U_m}|P_S) \leq  \frac{k-1}{k}D_k(P_{h(S;X)}\|P_{U_m}|P_S) + \frac{m}{k}.
    \end{align}
Again using \eqref{eq: D_k bound}, we obtain \eqref{eq: alpha = infty}.
\end{proof}
\begin{remark}
    The asymptotic behavior of $D_\infty$ in \eqref{eq: alpha = infty} was previously found in \cite{kavian2023output}, Theorem 3. Passing to asymptotics in the regime of $k=q^{\rho m}$ for a sufficiently
 small $\rho$ and letting $m\to\infty$, we can also obtain exponential decline of the divergence from our bounds.
\end{remark}

Next we present an LHL where the uniformity is measured by $\infty$-divergence. This result is useful for distilling outputs with high min-entropy. 

\begin{theorem}\label{thm: alpha = infty}
        Let $h : \cS \times \cX \to \zz_q^m$ be a $k^*$-universal hash function. Suppose $X$ is a random variable defined on $\cX$ and let $S \sim \cS$ be independent of $X$.  If $m \leq H_{k}(X) - \log_q(\frac{k}{2\epsilon \ln q})$ then 
    \begin{align}\label{eq: DC Infty}
        D_\infty(P_{h(S;X)}\|P_{U_m}|P_S) \leq \frac{m}{k} + \epsilon.
    \end{align}
\end{theorem}
This theorem follows immediately from \eqref{eq: infty bound} and \eqref{eq: D_k bound 2}.

A less explicit but more relaxed assumption that implies the same conclusion as \eqref{eq: DC Infty} is as follows:
 $$m \leq H_k(X)+ \log_q(\gamma(q^\epsilon)),$$ 
where $\gamma(\cdot)$ was defined in Remark \ref{rmk: gamma}. 

In summary, using a $k^*$-universal hash function with $k$ large enough compared to $m$ enables the generation of bit strings with high min-entropy that are nearly independent of the seed. 
A downside of this approach to generating uniform bits is its potential reliance on rather large seed lengths.

\subsection{Largest hash bucket}
Observe that $D_\infty(P_{h(S;X)} \| P_{U_m} | P_S)$ quantifies the probability of the most likely element produced by leftover hashing, averaged over all possible seeds. 
A closely related concept is the size of the ``largest hash bucket,'' where we estimate the maximum-frequency output element (also averaged over all seeds) when hashing all elements from a subset of $\cX$. 

To phrase this problem more formally, suppose that $\cA$ is a subset of $\cX$, and we apply a $k^*$-universal
hash function to every element of $\cA$. What is the (expected) size of the largest subset of $\cA$ on which
$h$ takes the same value? An upper bound for this quantity is given below. 

\begin{proposition}
    Let $k \in \{2,3\dots\}$. Let $h : \cS \times \cX \to \zz_q^m$ be $k^*$-universal. Let $S$ be a uniform random variable defined on $\cS$ and let $\cA$ be a subset of $\cX$. Then
    \begin{align*}
        \ee_{S}\big[\max_{u \in \zz_q^m}|\{x \in \cA: h(S,x) = u\}|\big] \leq \frac{kq^{\frac{m}{k}}}{\ln(kq^{m}/|\cA|+1)}.
    \end{align*}
\end{proposition}

Before we proceed to the proof, it is important to distinguish the setting we consider here from what was discussed in the previous section. 
In the earlier section, we hashed elements sampled from a random source, whereas in this section, we will hash all elements from a specific subset (without replacement). 
However, we will demonstrate that the latter case can also be modeled using an auxiliary random variable.
Thus, we can apply similar estimates from the previous section.

\begin{proof}
    To streamline our calculations, we introduce an additional random variable, $X$. Specifically, let $X$ be a uniform random variable defined on $\cA$, independent of $S$. We proceed as follows:
    \begin{align}
    \ee_{S}\big[&\max_{u \in \zz_q^m}|\{x \in \cA: h(S,x) = u\}|\big]\nonumber\\
        &= \sum_{s\in \cS}P_S(s)\max_{u \in \zz_q^m}\sum_{x\in \cA}\1\{h(s,x)=u\} \nonumber \\
        &=  \sum_{s\in \cS}P_S(s)\max_{u \in \zz_q^m}\sum_{x\in \cA}P_{h(S,X)|S,X}(u|s,x)\nonumber\\
        &=  |\cA|\sum_{s\in \cS}P_S(s)\max_{u \in \zz_q^m}\frac{1}{|\cA|}\sum_{x\in \cA}P_{h(S,X)|S,X}(u|s,x)\nonumber\\
        &=  |\cA|\sum_{s\in \cS}P_S(s)\max_{u \in \zz_q^m}P_{h(S,X)|S}(u|s)\nonumber\\
        &=  \frac{|\cA|}{q^m}\sum_{s\in \cS}P_S(s)\max_{u \in \zz_q^m}\frac{P_{h(S,X)|S,}(u|s)}{P_{U_m}(u)}\nonumber\\
        &=  \frac{|\cA|}{q^m}\sum_{s\in \cS}P_S(s)q^{D_\infty(P_{h(S,X)|S}(\cdot|s)\|P_{U_m})}\nonumber\\
        &\leq   \frac{|\cA|}{q^m}\sum_{s\in \cS}P_S(s)q^{\frac{k-1}{k}D_k(P_{h(S;X)|S}(\cdot|s)\|P_{U_m}) + \frac{m}{k}}
        \label{eq: temp1}\\
        &\leq   \frac{|\cA|}{q^m}q^{\frac{m}{k}}\Big[\sum_{s\in \cS}P_S(s)q^{(k-1)D_k(P_{h(S;X)|S}(\cdot|s)\|P_{U_m})}\Big]^{1/k}\label{eq: temp2},
    \end{align}
 where \eqref{eq: temp1} follows similarly to \eqref{eq: infty bound} and \eqref{eq: temp2}
is obtained by using Jensen's inequality.     Now observe that for any $\alpha \in (1,\infty),$
    \begin{align*}
        \sum_{s\in \cS}P_S(s)& q^{(\alpha-1)D_\alpha(P_{h(S,X)|S}(\cdot|s)\|P_{U_m|S}(\cdot|s))}\\ 
        & = \sum_{s\in \cS}P_S(s) \sum_{u \in \zz_q^m}\frac{P_{h(S,X)|S}(u|s)^\alpha}{P_{U_m|S}(u|s)^{\alpha-1}}\\
        & = \sum_{s\in \cS}\sum_{u \in \zz_q^m}\frac{P_{h(S,X),S}(u,s)^\alpha}{[P_{U_m}(u)P_S(s)]^{\alpha-1}}\\
        & = q^{(\alpha-1)D_\alpha(P_{h(S;X),S}\|P_{U_m}P_S)}
    \end{align*}
(cf. \eqref{eq: DE}).    
    Thus we can write \eqref{eq: temp2} as 
    \begin{align*}
         \ee_{S}\big[\max_{u \in \zz_q^m}&|\{x \in \cA: h(S,x) = u\}|\big] \\
        & \leq   \frac{|\cA|}{q^m}q^{\frac{m}{k}}q^{\frac{k-1}{k}D_k(P_{h(S;X),S}\|P_{U_m}P_S)}\\
        & \leq  \frac{|\cA|}{q^m}q^{\frac{m}{k}}\frac{kq^{m-H_k(X)}}{\ln(kq^{m-H_k(X)}+1)}\\
        & =  \frac{|\cA|}{q^m}q^{\frac{m}{k}}\frac{kq^{m-\log_q|\cA|}}{\ln(kq^{m-\log_q |\cA|}+1)}\\
        & =  \frac{kq^{\frac{m}{k}}}{\ln(kq^{m}/|\cA|+1)},
    \end{align*}
    where the second inequality follows from \eqref{eq: D_k bound}.
\end{proof}

The size of the largest hash bucket impacts the worst-case complexity of hash operations in practical settings. 
For example, the time complexity of lookups for a hashed element is proportional to the size of its hash bucket, with the worst-case search time dictated by the largest bucket. 
This section's results show using $k^*$-universal hash functions with sufficiently large $k$ 
can reduce the expected size of the largest hash bucket, thereby improving the worst-case efficiency of hash operations when averaged over all seeds. 

A common question regarding hash functions is: when hashing $N$ elements into $N$ buckets, what is the expected size of the largest hash bucket, averaged over all seeds? 
A well-known folklore result states that for universal hash functions, this size is $O(\sqrt{N})$. 
Another such result says that, if the $N$ elements are assigned uniformly and independently to $N$ buckets, the expected size of the largest bucket is $O(\ln N / \ln\ln N)$. 
A similar result holds for linear hash functions \cite{alon1999linear}, where the expected largest bucket size is  $O(\ln N \ln\ln N)$.

Letting $N = |\cA| = q^m$ in the last proposition, we obtain that the expected size of the largest hash bucket for $k^*$-universal hash functions is bounded by above $kN^{1/k}/\ln(k+1)$.  For $k=2$, this matches the aforementioned result for universal hash functions.
Moreover, when $k = m = \log_q N$, the behavior of $k^*$-universal hash functions closely approximates that of uniform and independent assignments in terms of the largest hash bucket. In other words, $k^\ast$-universality with $k=\log_q N$ matches 
the $N$-wise independent, i.e., fully random assignment in terms of the expected size of the largest hash bucket.

\section{Leftover hashing with side information}\label{sec: side}

In the literature on information-theoretic security, a common problem is to distill random, uniform bit strings that remain independent of any information accessible to adversarial parties \cite{bennett1988privacy, tyagi2023information}. For example, consider a scenario where we aim to extract a uniform distribution from a source $X$ while an adversary has access to a random variable $Z$ that is correlated with $X$. Our objective is to generate bits that are as close to uniform as possible and nearly independent of the adversary’s side information. This procedure is referred to as privacy amplification \cite{bennett1988privacy}. 

In this version of the problem, we aim to produce a strongly uniform random variable that is nearly independent of $Z$ in a strong sense. To meet these requirements, we adapt the statement of the LHL to the new setting by making slight modifications to previous theorems. Let us begin with defining the {\em conditional \Renyi entropy}\footnote{There are multiple versions of conditional \Renyi entropies. The one we use here is based on \cite{hayashi2011exponential}. For a more detailed account of conditional \Renyi entropies see \cite{iwamoto2013revisiting}.}. For $\alpha \in (1,\infty)$, it is given by
\begin{align*}
    H_\alpha(X|Z) & = \frac{1}{1-\alpha} \log_q\Big( \sum_{z \in \cZ}P_Z(z)\sum_{x \in \cX}P_{X|Z}(x|z)^\alpha \Big).
\end{align*}

We will now state several claims analogous to the earlier results, but additionally accounting for side information.
For all of them, we use identical assumptions, so rather than repeating them several times, we state them here.

\vspace*{.05in}{\em Assumptions $\text{\rm(XZS)}$}: {\em Let $X$ be a random variable supported on $\cX$ and let $Z$ be a random variable (possibly correlated with $X$) supported on a finite set $\cZ$. Let $S\sim \cS$ be a uniform random variable that is independent of both $X$ and $Z$.}
\vspace*{.05in}

This set of assumptions applies to all theorem-like statements in this section and will be suppressed below.

The following result forms an appropriate generalization of Theorem \ref{thm: main2}. It represents the most
general form of the LHL with side information that we claim, so it is stated first. Its proof follows the lines 
of the proof of Theorem~\ref{thm: main2}.
\begin{theorem}\label{thm: side main}
     Let $k \in \{2,3\dots\}$ and $\alpha \in (1,k]$.  Let $h : \cS \times \cX \to \zz_q^m$ be $k^*$-universal, then
    \begin{align}\label{eq: side continuous}
        &q^{(\alpha-1)D_\alpha(P_{h(S;X),S,Z}\|P_{U_m}P_SP_Z)} \nonumber\\
        & \qquad \leq  \sum_{l=1}^{\lceil \alpha \rceil-1} l \stirlingII{\lceil \alpha \rceil-1}{l} q^{(\alpha-l)(m-H_{\alpha}(X|Z)} \nonumber\\
        & \qquad \quad \quad + \sum_{l=1}^{\lceil \alpha \rceil}\stirlingII{\lceil \alpha\rceil-1}{l-1} q^{(\lceil \alpha \rceil-l)(m-H_{ \alpha }(X|Z)}. 
    \end{align}
\end{theorem}

We limit ourselves to a proof sketch since the argument closely follows the proof of Theorem \ref{thm: main2}.
A straightforward calculation gives:
\begin{align}
    & q^{(\alpha-1)D_\alpha(P_{h(S;X),S,Z}\|P_{U_m}P_SP_Z)} \nonumber \\
    &\quad=  \sum_{z \in \cZ}P_Z(z)q^{m(\alpha-1)}\sum_{s\in\cS}P_S(s)\nonumber\\
    &\quad\quad \times \sum_{u\in\zz_q^m}
          \prod_{i=1}^{k-1}\Big[\sum_{x_i \in \cX}\1\{h(s,x_i) = u\}P_{X|Z}(x_i|z)\Big]\nonumber \\ 
    & \quad \quad \quad \times\Big[\sum_{x_{k} \in \cX}\1\{h(s,x_k) = u\}P_{X|Z}(x_{k}|z)\Big]^{\alpha-k+1} \label{eq: alpha expand 2}
\end{align}

For each fixed $z \in \cZ$, we can bound the inner sums as in the proof of Theorem \ref{thm: main2}, with the probability mass function $P_X$ replaced by $P_{X|Z}(\cdot|z)$. Averaging over $z \in \cZ$, and applying Jensen’s inequality appropriately, we obtain the desired result.

Similarly, we can obtain a generalized version of Theorem \ref{thm: main 0} that includes the side information term.
\begin{theorem}\label{thm: side main 0}
    Let $\epsilon >0 $ and let $k \in \{2,3\dots\}$ and $\alpha \in \{2,3,\dots,k\}$. 
   Let $h : \cS \times \cX \to \zz_q^m$ be $k^*$-universal. If $m \leq H_{\alpha}(X|Z) -  \log_q(\frac{\alpha^2}{2\epsilon(\alpha-1)\ln q})$, then 
    \begin{align}
        D_\alpha(P_{h(S;X),S,Z}\|P_{U_m}P_SP_Z) \leq \epsilon.
    \end{align}
\end{theorem}

In its turn, Corollary \ref{cor: LHL (1,2]} extends to the following statement, which is due to 
\cite{hayashi2016equivocations}.
\begin{corollary}\label{thm: side LHL (1,2]}
    Let $\epsilon > 0$ and $\alpha \in (1,2]$. Let $h : \cS \times \cX \to \zz_q^m$ be $2$-universal.  If $m \leq H_{\alpha}(X|Z)-\frac{1}{\alpha-1}\log_q (\frac{1}{\epsilon \ln q (\alpha-1)})$. Then
    \begin{align}\label{eq: side LHL (1,2]}
        D_\alpha(P_{h(S;X),S,Z}\|P_{U_m}P_SP_Z) \leq \epsilon.
    \end{align}
\end{corollary}

Plainly, the results presented in Section \ref{sec: min} can also be adjusted so that they incorporate side information. As expected, in this case too, the proximity measure must be weakened to achieve the uniformity guarantees.

\begin{proposition}\label{prop: side alpha > k 2}
     Let $k \in \{2,3\dots\}$. Let $h : \cS \times \cX \to \zz_q^m$ be $k^*$-universal. Then for $\alpha \in (k,\infty)$
    \begin{align}\label{eq: side alpha > k 2}  
        D_\alpha&(P_{h(S;X)} \|P_{U_m}|P_{SZ})\nonumber\\
        &\leq \frac{\alpha-k}{k(\alpha-1)}m + \frac{\alpha}{\alpha-1}\log_q \Big(\frac{kq^{m-H_k(X|Z)}}{\ln(kq^{m-H_k(X|Z)}+1)}\Big),    
    \end{align}
    and
    \begin{align}\label{eq: side alpha = infty 2}
         D_\infty&(P_{h(S;X)}\|P_{U_m}|P_{SZ})\nonumber\\ &\leq \frac{m}{k} + \log_q\Big(\frac{kq^{m-H_k(X|Z)}}{\ln(kq^{m-H_k(X|Z)}+1)}\Big).
    \end{align}
\end{proposition}

We can also state and prove a result analogous to Theorem \ref{thm: alpha = infty}.

\begin{theorem}\label{thm: side alpha = infty}
        Let $h : \cS \times \cX \to \zz_q^m$ be a $k^*$-universal hash function. If $m \leq H_{k}(X|Z) - \log_q\big(\frac{k}{2\epsilon \ln q}\big)$ then 
    \begin{align}
        D_\infty(P_{h(S;X)}\|P_{U_m}|P_{SZ}) \leq \frac{m}{k} + \epsilon.
    \end{align}
\end{theorem}

\section{Concluding remarks}

In this work, we have established uniformity guarantees for $k^*$-universal hash functions. 
Specifically, we show that for $\alpha \in (1, k]$, both uniformity and independence properties can be ensured using the $\alpha$-\Renyi divergence, extending the previous results in \cite{hayashi2016equivocations} to all $\alpha\in (2,k]$.
In particular, we demonstrate that it is possible to extract nearly all of the $\alpha$-entropy of the source to generate uniform bits. For $\alpha > k$, we derive uniformity guarantees based on conditional \Renyi divergence. 
This conditional version provides the same uniformity guarantees as the unconditional even though it relies on  
a less stringent version of independence between the seed and the extracted bits. In particular, we provide 
min-entropy guarantees for $k^*$-universal hash functions and estimate the size of the largest hash bucket, a key factor in the worst-case performance of hash operations. Finally, we extend these uniformity guarantees to scenarios where uniform bits need to be distilled while ensuring independence from an adversary’s accessible information. This result strengthens security guarantees in cryptographic applications such as secret key generation.

The seed lengths required for $k^*$-universal hash functions can be rather large, so it is of interest 
to explore R{\'e}nyi-divergence based uniformity guarantees for extractors with shorter seed lengths.
An obstacle to this has been pointed in the literature, namely \cite[p.205]{Vadhan2012Pseudo}
implies that 2-\Renyi extractors that convert nearly all of the 2-\Renyi entropy into random bits, require seed length of at least $\min(\log_q|\cX|-m, m/2)-O(1)$. In comparison, 1-\Renyi extractors 
are capable of constructing nearly $H(X)$ random bits with an optimal seed length of $O(\log_q\log_q|\cX|)$, much
shorter than the known families of $k^*$-universal hash functions.
It is unclear to us whether $k^*$-universal hash functions can match the optimal seed lengths of $\alpha$-\Renyi extractors for $\alpha>1$.

In conclusion we note that the results of this work can be extended to almost $k^*$-universal hash functions \cite[p.~77]{tyagi2023information}, defined by replacing \eqref{eq: k-hash} with a relaxed upper bound of $O(q^{-m(k-1)})$. Our results can also be extended to smoothed versions of the 
 \Renyi entropy mentioned briefly in Proposition~\ref{thm: smoothed} above.

\appendix

\section{Proof of Theorem \ref{thm: main2}}\label{Appendix A}

The case of $\alpha = k$ was already established in Theorem \ref{thm: main1}, so it remains to extend its result to all non-integer values of $\alpha$ within the interval $(1, k]$. It is sufficient to establish it for $\alpha \in (k-1, k]$ since 
a $k^*$-universal hash function is also $l^*$-universal for every $l \in {2, \dots, k}$. In summary, we need prove the following.

\begin{proposition}\label{prop: main_k}
     Let $k \in \{2,3\dots\}$ and $\alpha \in (k-1,k]$. Let $X$ be a random variable defined on $\cX$ and let $S\sim \cS$ be a uniform random variable that is independent of $X$. Let $h : \cS \times \cX \to \zz_q^m$ be $k^*$-universal. Then
    \begin{align}
        q^{(\alpha-1)D_\alpha(P_{h(S;X),S}\|P_{U_m}P_S)} 
        & \leq \sum_{l=1}^{k-1}l\stirlingII{k-1}{l} q^{(\alpha-l)(m-H_{\alpha}(X))} \nonumber\\
        & + \sum_{l = 1}^{k}\stirlingII{k-1}{l-1}q^{(k-l)(m-H_{\alpha}(X))}. 
    \end{align}
\end{proposition}
\begin{proof}
We begin with a sequence of straightforward calculations similar to the ones that led to \eqref{eq: u0} and \eqref{eq: u1}.
As before, we will use the notation $P_{X^{k-1}}(x^{k-1}):=P_X(x_1)\dots P_X(x_{k-1})$. Note that we isolate the last
term $x_k$ into a separate sum. We have
\begin{align}
     &q^{(\alpha-1)D_\alpha(P_{h(S;X),S}\|P_{U_m}P_S)}= q^{m(\alpha-1)}\sum_{s}P_S(s) \nonumber\\
    &\quad \times\sum_{u}\prod_{i=1}^{k-1}\Big[\!\sum_{x_i \in \cX}\!\1\{h(s,x_i) = u\}P_X(x_i)\Big]\nonumber\\
    &\qquad \times \Big[\sum_{x_{k} \in \cX}\1\{h(s,x_k) = u\}P_X(x_{k})\Big]^{\alpha-k+1} \nonumber\\
    & \quad = q^{m(\alpha-1)}\sum_{u}\sum_{s}P_S(s)\sum_{x^{k-1} \in \cX^{k-1}} P_{X^{k-1}}(x^{k-1})\nonumber \\ 
    & \qquad \times \1\{h(s,x_1) = \dots = h(s,x_{k-1}) = u\}\nonumber \\
    & \qquad \times \Big[\!\!\sum_{x_{k} \in \{x_1,\dots,x_{k-1}\}}\hspace*{-.3in}\1\{h(s,x_k) = u\}P_X(x_{k}) \nonumber\\
    & \qquad + \hspace*{-.3in}\sum_{x_{k} \in \cX\setminus \{x_1,\dots,x_{k-1}\}}\hspace*{-.3in}\1\{h(s,x_k) = u\}P_X(x_{k})\Big]^{\alpha-k+1}.
    \label{eq: alpha expand}
\end{align}
Let $a\in\rr_+^n$ and $0<\beta\le1$. From the monotonicity of $\ell_p$-norms, $\|a\|_1\le \|a\|_\beta$, or
\begin{align}\label{eq: beta power}
    (a_1+\dots+a_n)^\beta \leq a_1^\beta+\dots+a_n^\beta.
\end{align}
Noting that $0<\alpha-k+1\le 1$ and using \eqref{eq: beta power} (with $n=2$) in \eqref{eq: alpha expand} , we can write it as follows:
\begin{align}\label{eq: A1 + A2}
    q^{(\alpha-1)D_\alpha(P_{h(S;X),S}\|P_{U_m}P_S)} \leq A_1 + A_2,
\end{align}
where 
\begin{align*}
    A_1 & = q^{m(\alpha-1)}\sum_{u}\sum_{s}P_S(s) \sum_{x^{k-1} \in \cX^{k-1}} P_{X^{k-1}}(x^{k-1})\nonumber\\
        & \quad \times \1\{h(s,x_1) = \dots = h(s,x_{k-1}) = u\} \nonumber \\
        & \quad \times\Big(\sum_{x_{k} \in \{x_1,\dots,x_{k-1}\}}\1\{h(s,x_k) = u\}P_X(x_{k})\Big)^{\alpha-k+1}\
\end{align*}
and
\begin{align*}
    A_2 & = q^{m(\alpha-1)}\sum_{u}\sum_{s}P_S(s) \sum_{x^{k-1} \in \cX^{k-1}}  P_{X^{k-1}}(x^{k-1})\nonumber\\
    & \quad \times \1\{h(s,x_1) = \dots = h(s,x_{k-1}) = u\}\nonumber\\
& \quad \times\Big(\sum_{x_{k} \in \cX\setminus \{x_1,\dots,x_{k-1}\}}\1\{h(s,x_k) = u\}P_X(x_{k})\Big)^{\alpha-k+1}.  
\end{align*}

We now proceed to bound $A_1$ and $A_2$ separately. During this process, we will rearrange the order of summations, similar to the approach used in the proof of Theorem \ref{thm: main1}. We again work with partitions of the 
variables, although unlike the proof of Theorem \ref{thm: main1}, we now define $T$ as a partition of $\cX^{k-1}$ rather than $\cX^k$.

\vspace*{.1in}{\em Bounding $A_1$.}

\vspace*{.05in}
We will use elements of notation introduced in the proof of Theorem~\ref{thm: main1}. Given
a partition $\cP\in{\mathfrak P}_{k-1},$ let us
use \eqref{eq: beta power} for the last bracket in the expression for $A_1$ to obtain
\begin{align}
    A_1 & \leq q^{m(\alpha-1)}\sum_{u}\sum_{s}P_S(s) \sum_{x^{k-1} \in \cX^{k-1}} P_{X^{k-1}}(x^{k-1}) \nonumber \\
        & \quad \times \1\{h(s,x_1) = \dots = h(s,x_{k-1}) = u\} \nonumber \\
        & \quad \times \!\sum_{x_{k} \in \{x_1,\dots,x_{k-1}\}}\1\{h(s,x_k) = u\}^{\alpha-k+1}P_X(x_{k})^{\alpha-k+1}\nonumber \\[.1in]
        & = q^{m(\alpha-1)}\sum_{u}\sum_{s}P_S(s) \sum_{x^{k-1} \in \cX^{k-1}} P_{X^{k-1}}(x^{k-1})\nonumber\\
        & \quad \times \sum_{x_{k} \in \{x_1,\dots,x_{k-1}\}}P_X(x_{k})^{\alpha-k+1}\nonumber \nonumber \\
        & \quad \times  \1\{h(s,x_1) = \dots = h(s,x_{k}) = u\}\nonumber \\[.1in]
        & = q^{m(\alpha-1)}\sum_{u} \sum_{x^{k-1} \in \cX^{k-1}} P_{X^{k-1}}(x^{k-1})\nonumber \\
        & \quad \times \sum_{x_{k} \in \{x_1,\dots,x_{k-1}\}}P_X(x_{k})^{\alpha-k+1}\nonumber \\
        & \quad \times \Pr_{S}(h(S,x_1) = \dots = h(S,x_{k}) = u)\nonumber \\[.1in]
        & = q^{m(\alpha-1)} \hspace{-8 pt} \sum_{x^k \in \cX^{k-1}} \hspace{-8 pt} P_{X^{k-1}}(x^{k-1})\hspace{-8 pt}\sum_{x_{k} \in \{x_1,\dots,x_{k-1}\}} \hspace{-8 pt} P_X(x_{k})^{\alpha-k+1}\nonumber \\
        & \quad \times \Pr_{S}(h(S,x_1) = \dots = h(S,x_{k}) )\nonumber \\[.1in]
        & \leq q^{m(\alpha-1)} \sum_{\cP \in \mathfrak P_{k-1}}\sum_{x^{k-1} \in T_\cP} P_{X^{k-1}}(x^{k-1}) \nonumber \\
        & \quad \times  \sum_{x_{k} \in \{x_1,\dots,x_{k-1}\}}\hspace*{-.2in}P_X(x_{k})^{\alpha-k+1}q^{-m(\eta(x^{k-1})-1)},\label{eq: A_1 b1}
\end{align}
where as before in the proof of Theorem~\ref{thm: main1}, $\eta(\cdot)$ is the number of distinct entries in the argument tuple.
Now let us fix an $l$-partition $\cP \in {\mathfrak P}_{k-1}(l),1\le l\le k-1$ and recall the notation $p_i,\pi_i$ for the size of the $i$th block and for its element, respectively. Observe that in this case $\eta(x^{k-1}) = l$. 
Let us bound the following sum:
\begin{align}
     \sum_{x^{k-1} \in T_\cP} &P_{X^{k-1}}(x^{k-1}) \sum_{x_{k} \in \{x_1,\dots,x_{k-1}\}}P_X(x_{k})^{\alpha-k+1}\nonumber \\
    &= \sum_{x^{k-1} \in T_\cP} \prod_{j=1}^l P_X(x_{\pi_j})^{p_j} \sum_{i = 1}^l P_X(x_{\pi_i})^{\alpha -k +1 } \nonumber\\
    &= \sum_{i = 1}^l \sum_{x^{k-1} \in T_\cP} \prod_{j=1}^l P_X(x_{\pi_j})^{q_j(i)}  \label{eq: A_1 b20}, 
\end{align}
where 
\begin{align*}
        q_j(i) =
\begin{cases}
		p_j + (\alpha -k +1) & \text{ if }  i = j \\
		p_j & \text{ otherwise .}	
\end{cases}
\end{align*}

This exponent appears in \eqref{eq: A_1 b20} because we lump together the probabilities of the elements of the $j$th
block and the added element $x_k$ which falls in this block. 
 
By relaxing the requirement that variables in different blocks must be distinct,
\begin{align}
      \sum_{x^{k-1} \in T_\cP} &P_{X^{k-1}}(x^{k-1}) \sum_{x_{k} \in \{x_1,\dots,x_{k-1}\}}P_X(x_{k})^{\alpha-k+1}\nonumber \\
    &= \sum_{i = 1}^l \sum_{x^{k-1} \in T_\cP} \prod_{j=1}^l P_X(x_{\pi_j})^{q_j(i)} \nonumber\\
    &\leq \sum_{i = 1}^l \prod_{j=1}^l \sum_{x \in \cX}  P_X(x)^{q_j(i)}. \label{eq: A_1 b2}
\end{align}

Observe that for any $i$, $\sum_jq_j(i) = \alpha$ and $q_j(i) \geq 1$. Now let us fix $i,j$ and evaluate the sum on $x$ in \eqref{eq: A_1 b2}. If $q_j(i) = 1$, this sum is equal to 1. Otherwise, applying 
H{\"o}lder's inequality as in \eqref{eq: Holder}, we obtain 
\begin{align}\label{eq: Holder3}
    \sum_{x \in \cX}P_X(x)^{q_j(i)} \leq \Big(\sum_{x \in \cX}P_X(x)^{\alpha-l+1}\Big)^{\frac{q_j(i)-1}{\alpha-l}}.
\end{align}
Arguing as in \eqref{eq: prob bound}, we now obtain
\begin{align}\label{eq: A_1 b3}
    \prod_{j=1}^l \sum_{x \in \cX}  P_X(x)^{q_j(i)}  \leq q^{-(\alpha-l)H_{\alpha-l+1}(X)}.
\end{align}

Finally, let us substitute \eqref{eq: A_1 b2} and \eqref{eq: A_1 b3} into \eqref{eq: A_1 b1} to obtain
a bound for $A_1$:
\begin{align}
    A_1 
    & \leq q^{m(\alpha-1)}\sum_{l=1}^{k-1} \sum_{\cP \in \mathfrak P_{k-1}(l)} \sum_{i = 1}^lq^{-(\alpha-l)H_{\alpha-l+1}(X)} q^{-m(l-1)} \nonumber \\
    & = \sum_{l=1}^{k-1}\stirlingII{k-1}{l}\sum_{i = 1}^lq^{-(\alpha-l)H_{\alpha-l+1}(X)+m(\alpha-l)}\nonumber \\
    & = \sum_{l=1}^{k-1}l\stirlingII{k-1}{l} q^{(\alpha-l)(m-H_{\alpha-l+1}(X))}\nonumber \\
    & \leq  \sum_{l=1}^{k-1}l\stirlingII{k-1}{l} q^{(\alpha-l)(m-H_{\alpha}(X))}\label{eq: A_1 b4},
\end{align}
where the final inequality follows from the monotonicity of \Renyi entropy.

\vspace*{.1in}{\em Bounding $A_2$.}

\vspace*{.05in} We will rearrange the sums in $A_2$. To shorten the writing, let $\1_{(k-1)}(u):=\1\{h(s,x_1) = \dots = h(s,x_{k-1}) = u\}$. With this, we have
\begin{align}
    A_2 &= q^{m(\alpha-1)}\sum_{u}\sum_{s}P_S(s)\nonumber\\
    & \quad \times \sum_{l=1}^{k-1}\sum_{\cP \in \mathfrak P_{k-1}(l)}\sum_{x^{k-1} \in T_{\cP}} P_{X^{k-1}}(x^{k-1})\1_{(k-1)}(u)\nonumber\\
    & \quad \times 
        \Big[\sum_{x_{k} \in \cX\setminus \{x_1,\dots,x_{k-1}\}} 
        \1\{h(s,x_k) = u\}P_X(x_{k})\Big]^{\alpha-k+1}\nonumber\\
    &= \sum_{l=1}^{k-1}\!\sum_{\cP \in \mathfrak P_{k-1}(l)}\!\sum_{u} q^{m(\alpha-1)}\hspace{-8 pt}
\sum_{x^{k-1} \in T_{\cP}} \hspace{-8 pt}P_{X^{k-1}}(x^{k-1})\sum_{s}P_S(s)\nonumber\\
    &  \times  \1_{(k-1)}(u)
        \Big[\hspace*{-8 pt}\sum_{ \hspace*{-6 pt} x_{k} \in \cX\setminus \{x_1,\dots,x_{k-1}\}} \hspace{-12 pt}
        \1\{h(s,x_k) = u\}P_X(x_{k})\Big]^{\alpha-k+1}\nonumber\\
    &=\sum_{l=1}^{k-1} \sum_{\cP \in \mathfrak P_{k-1}(l)}B_{\cP}\label{eq: A2 bound},
\end{align}
where we have denoted the sum on $u$ by $B_\cP$. Let us fix a partition $\cP\in{\mathfrak P}_{k-1}$ and let $l=|\cP|$. Further, define 
\begin{gather}
    W(u,x^{k-1},s) := q^{m(\alpha-1)}P_{X^{k-1}}(x^{k-1})P_S(s)\1_{(k-1)}(u),  \\
        C_{\cP}  :=  \sum_{u}\sum_{x^{k-1} \in T_{\cP}}\sum_{s}W(u,x^{k-1},s). \label{eq: CP}
\end{gather}
Let us bound $C_{\cP}$ from above. First, we rewrite it as follows: 
\begin{align}
    C_{\cP} 
    & = q^{m(\alpha-1)}\sum_{u}\sum_{x^{k-1} \in T_{\cP}} P_{X^{k-1}}(x^{k-1}) \sum_{s}P_S(s)\1_{(k-1)}(u) \nonumber \\
    & = q^{m(\alpha-1)}\sum_{x^{k-1} \in T_{\cP}} P_{X^{k-1}}(x^{k-1}) \nonumber \\
    & \quad \times \sum_{u}\Pr_S\big(h(S,x_1) = \dots = h(S,x_{k-1}) = u \big)\nonumber \\
    & = q^{m(\alpha-1)}\sum_{x^{k-1} \in T_{\cP}} P_{X^{k-1}}(x^{k-1}) \nonumber\\
    & \quad \times \Pr_S\big(h(S,x_1) = \dots = h(S,x_{k-1}) \big)\nonumber \\
    & \leq q^{m(\alpha-1)}\sum_{x^{k-1} \in T_{\cP}} P_{X^{k-1}}(x^{k-1}) q^{-m(l-1)}\nonumber\\
    & = q^{m(\alpha-l)}\sum_{x^{k-1} \in T_{\cP}} P_{X^{k-1}}(x^{k-1}). \label{eq: C_p exp}
\end{align}
Now use \eqref{eq: prob bound} with $k$ replaced with $k-1$:
    \begin{align}\label{eq: C_p exp2}
        C_{\cP} \leq  q^{m(\alpha-l)}q^{-(k-1-l)H_{k-l}(X)}.
    \end{align}
Next, return to bounding $B_\cP$:
\begin{align}
    B_{\cP} 
    & = \sum_{u}\sum_{x^{k-1} \in T_{\cP}}\sum_{s}W(u,x^{k-1},s)\nonumber\\
    & \quad \times \Big[\hspace*{-.2in}\sum_{\hspace*{.2in}x_{k} \in \cX\setminus \{x_1,\dots,x_{k-1}\}} \hspace*{-.2in}
        \1\{h(s,x_k) = u\}P_X(x_{k})\Big]^{\alpha-k+1} \nonumber\\
    & = C_{\cP}\sum_{u}\sum_{x^{k-1} \in T_{\cP}}\sum_{s}\frac{W(u,x^{k-1},s)}{C_{\cP}}\nonumber\\
    & \quad \times\Big[\hspace*{-.2in}\sum_{\hspace*{.2in}x_{k} \in \cX\setminus \{x_1,\dots,x_{k-1}\}} \hspace*{-.2in}
        \1\{h(s,x_k) = u\}P_X(x_{k})\Big]^{\alpha-k+1}\nonumber\\
    & \leq C_{\cP}\Big(\sum_{u}\sum_{x^{k-1} \in T_{\cP}} \sum_{s}\frac{W(u,x^{k-1},s)}{C_{\cP}}\nonumber\\
    & \quad \times\sum_{x_{k} \in \cX\setminus \{x_1,\dots,x_{k-1}\}} \hspace*{-.2in} \1\{h(s,x_k) = u\}P_X(x_{k})\Big)^{\alpha-k+1},\label{eq: Bp bound}
\end{align}
where the last expression is obtained from the concavity of the function $z^{\alpha-k+1}, z>0$. Indeed, note that 
$0<\alpha-k+1\le 1$ and the weights form a probability vector by \eqref{eq: CP}, so Jensen's inequality applies.
Simplifying the expression in the parentheses, we further obtain
\begin{align}
    & \sum_{u}\sum_{x^{k-1} \in T_{\cP}}\sum_{s}\frac{W(u,x^{k-1},s)}{C_{\cP}}\nonumber\\
    & \quad \times \sum_{x_{k} \in \cX\setminus \{x_1,\dots,x_{k-1}\}} \hspace*{-.2in} \1\{h(s,x_k) = u\}P_X(x_{k})\nonumber\\
    &= \frac{q^{m(\alpha-1)}}{C_{\cP}}\sum_{u}\sum_{x^{k-1} \in T_{\cP}}
    \sum_{x_{k} \in \cX\setminus \{x_1,\dots,x_{k-1}\}}\hspace*{-.2in} P_{X^k}(x^k) \nonumber\\
    & \quad \times \sum_{s}P_S(s)\1\{h(s,x_1) = h(s,x_k) = u\} \nonumber\\
    &= \frac{q^{m(\alpha-1)}}{C_{\cP}}\sum_{x^{k-1} \in T_{\cP}}\sum_{x_{k} \in \cX\setminus \{x_1,\dots,x_{k-1}\}}\hspace*{-.2in} P_{X^k}(x^k) \nonumber\\
    & \quad \times \Pr_S\big(h(S,x_1) = \dots = h(S,x_k)\big) \nonumber\\
    & \leq \frac{q^{m(\alpha-1)}}{C_{\cP}}\sum_{x^{k-1} \in T_{\cP}}\sum_{x_{k} \in \cX\setminus \{x_1,\dots,x_{k-1}\}}\hspace*{-.2in} P_{X^k}(x^k) q^{-m(\eta(x^{k})-1)} \nonumber\\
    &\leq  \frac{q^{m(\alpha-1)}}{C_{\cP}}\sum_{x^{k-1} \in T_{\cP}} P_{X^{k-1}}(x^{k-1}) q^{-ml} \nonumber\\
    &\leq  \frac{q^{m(\alpha-l-1)}q^{-(k-l-1)H_{k-l}(X)}}{C_\cP}, \label{eq: temp3}
\end{align}
where on the third-to-last line we used the definition of the $k^*$-universal hash function, and 
where the last inequality follows upon substituting for $P_{X^{k-1}}(x^{k-1})$ as in \eqref{eq: prob bound}.
Using \eqref{eq: temp3} in \eqref{eq: Bp bound}, we obtain 
\begin{align*}
    B_{\cP} 
    & \leq C_{\cP}\Big( \frac{q^{m(\alpha-l-1)}q^{-(k-l-1)H_{k-l}(X)}}{C_\cP}\Big)^{\alpha-k+1} \\
    & = C_{\cP}^{k-\alpha} \big(q^{m(\alpha-l-1)}q^{-(k-l-1)H_{k-l}(X)}\big)^{\alpha-k+1}.
\end{align*}
Applying \eqref{eq: C_p exp2},
\begin{align*}
    B_{\cP} 
    & \leq \big(q^{m(\alpha-l)}q^{-(k-l-1)H_{k-l}(X)}\big)^{k-\alpha}\\
    & \quad \times
    \big(q^{m(\alpha-l-1)}q^{-(k-l-1)H_{k-l}(X)}\big)^{\alpha-k+1}\\
    & = q^{(k-l-1)(m-H_{k-l}(X))}.
\end{align*}
Substituting this back to \eqref{eq: A2 bound},
\begin{align}
    A_2
    & = \sum_{l=1}^{k-1}\stirlingII{k-1}{l}q^{(k-1-l)(m-H_{k-l}(X))}\nonumber\\
    & \leq \sum_{l = 1}^{k-1}\stirlingII{k-1}{l}q^{(k-1-l)(m-H_{\alpha}(X))}\nonumber\\
    & = \sum_{l = 1}^{k}\stirlingII{k-1}{l-1}q^{(k-l)(m-H_{\alpha}(X))}\label{eq: A_2 b_2},
\end{align}
where the second inequality follows from the monotonicity of $H_\alpha(\cdot)$ on $\alpha$, and the last step \eqref{eq: A_2 b_2} uses $\stirlingII{k-1}{0}=0$.
Now using the estimates \eqref{eq: A_1 b4} and \eqref{eq: A_2 b_2} in \eqref{eq: A1 + A2} completes the proof.
\end{proof}

\bibliographystyle{IEEEtranS}
\bibliography{smoothing}

\vfill
\end{document}